\documentclass[12pt,american]{article}
\usepackage[latin9]{inputenc}
\usepackage{color}
\usepackage{babel}
\usepackage{units}
\usepackage{mathrsfs}
\usepackage{url}
\usepackage{amsmath}
\usepackage{amsthm}
\usepackage{amssymb}
\usepackage[unicode=true,
 bookmarks=true,bookmarksnumbered=true,bookmarksopen=false,
 breaklinks=true,pdfborder={0 0 1},backref=false,colorlinks=true]
 {hyperref}
\hypersetup{pdftitle={Heat Kernel Estimates for Fractional Heat Equation},
 pdfauthor={Jos{\'e}Lu{\'\i}sdaSilva},
 pdfsubject={Heat Kernel Estimates for Fractional Heat Equation},
 pdfkeywords={General fractional derivative, Subordination, Heat equation, Long time behavior },
 linkcolor=black,citecolor=black,urlcolor=black,filecolor=black}

\makeatletter

\DeclareTextSymbolDefault{\textquotedbl}{T1}

\theoremstyle{plain}
\newtheorem{thm}{\protect\theoremname}
\theoremstyle{definition}
\newtheorem{example}[thm]{\protect\examplename}
\theoremstyle{plain}
\newtheorem{lem}[thm]{\protect\lemmaname}
\theoremstyle{remark}
\newtheorem{rem}[thm]{\protect\remarkname}

\usepackage{xcolor}
\usepackage{babel}
\usepackage{units}
\usepackage{mathrsfs}
\usepackage{url}
\PassOptionsToPackage{normalem}{ulem}
\usepackage{ulem}

\usepackage{pdfsync}
\numberwithin{equation}{section}

\makeatother

\providecommand{\examplename}{Example}
\providecommand{\lemmaname}{Lemma}
\providecommand{\remarkname}{Remark}
\providecommand{\theoremname}{Theorem}

\begin{document}
\title{Heat Kernel Estimates for Fractional Heat Equation}
\author{\textbf{Anatoly N. Kochubei}\\
 Institute of Mathematics,\\
 National Academy of Sciences of Ukraine, \\
 Tereshchenkivska 3, \\
 Kyiv, 01024 Ukraine\\
 Email: kochubei@imath.kiev.ua\and\textbf{Yuri Kondratiev}\\
 Department of Mathematics, University of Bielefeld, \\
 D-33615 Bielefeld, Germany,\\
 Dragomanov University, Kiev, Ukraine\\
 Email: kondrat@mathematik.uni-bielefeld.de\and\textbf{Jos{\'e}
Lu{\'i}s da Silva},\\
 CIMA, University of Madeira, Campus da Penteada,\\
 9020-105 Funchal, Portugal.\\
 Email: joses@staff.uma.pt}
\date{\today}

\maketitle
 
\begin{abstract}
We study the long-time behavior of the Cesaro means of fundamental
solutions for fractional evolution equations corresponding to random
time changes in the Brownian motion and other Markov processes. We
consider both stable subordinators leading to equations with the Caputo-Djrbashian
fractional derivative and more general cases corresponding to differential-convolution
operators, in particular, distributed order derivatives. 
\end{abstract}
\tableofcontents{}

\section{Introduction}

Let $\{X_{t},t\geq0;P_{x},x\in E\}$ be a strong Markov process in
a phase space $\mathbb{R}^{d}$. Denote by $T_{t}$ its transition
semigroup (in an appropriate Banach space) and by $L$ the generator
of this semigroup. Let $S_{t},t\geq0$ be a subordinator (i.e., a
non-decreasing real-valued L{\'e}vy process) with $S_{0}=0$ and
the Laplace exponent $\Phi$: 
\[
\mathbf{E}[e^{-\lambda S_{t}}]=e^{-t\Phi(\lambda)}\;\;t,\lambda>0.
\]
We assume that $S_{t}$ is independent of $X_{t}$.

Denote by $E_{t},t>0$ the inverse subordinator, and introduce the
time-changed process $Y_{t}=X_{E_{t}}$. A general aim is to analyze
the properties of $Y_{t}$ depending on the given Markov process $X_{t}$
and the particular choice of subordinator $S_{t}$. There is interest
in this kind of problem in diverse disciplines. In addition to purely
stochastic motivations, a transform of the Markov process $X_{t}$
in the non-Markov one $Y_{t}$ implies the presence of effects in
the corresponding dynamics. This feature is important in a number
of physical models. In particular, progress in the understanding of
this process may lead to the realization of useful models of biological
time in the evolution of species and ecological systems. Currently,
there exist rather complete studies of such problems in the case of
so-called stable subordinators \cite{BM01,MS2004} and in special
examples of initial processes $X_{t}$ (see, e.g., \cite{Magdziarz2015},
\cite{Li2007}, \cite{Mimica2016}).

As a basic characteristic of the new process $Y_{t}$, we may study
the time evolution 
\[
u(t,x)=\mathbf{E}^{x}[f(Y_{t})]
\]
for a given initial data $f$.

As it was pointed out in several works, see e.g. \cite{Toaldo2015},
\cite{Chen2017} and references therein, $u(t,x)$ is the unique strong
solution (in some appropriate sense) to the following Cauchy problem
\begin{equation}
\mathbb{D}_{t}^{(k)}u(t,x)=Lu(t,x)\;\;u(0,x)=f(x).\label{FCP}
\end{equation}
Here we have a generalized fractional derivative (GFD for short) 
\[
\mathbb{D}_{t}^{(k)}\phi(t)=\frac{d}{dt}\int_{0}^{t}k(t-s)(\phi(s)-\phi(0))ds
\]
with a kernel $k$ uniquely defined by $\Phi$.

Let $u_{0}(t,x)$ be the solution to a similar Cauchy problem but
with ordinary time derivative 
\begin{equation}
\frac{\partial}{\partial t}u(t,x)=Lu(t,x)\;\;u(0,x)=f(x).\label{CP}
\end{equation}
In stochastic terminology, it is the solution to the forward Kolmogorov
equation corresponding to the process $X_{t}$. Under quite general
assumptions there is a convenient and essentially obvious relation
between these evolutions that is known as the subordination principle:
\[
u(t,x)=\int_{0}^{\infty}u_{0}(\tau,x)G_{t}(\tau)d\tau,
\]
where $G_{t}(\tau)$ is the density of $E_{t}$.

A similar relation holds for fundamental solutions (or heat kernels
in another terminology) $v(x,t)$ and $v^{E}(x,t)$ of equations (\ref{CP})
and (\ref{FCP}), respectively. For certain classes of \emph{a priori}
bounds for fundamental solutions $v(x,t)$, the properties of the
subordinated kernels were studied in \cite{Chen2018}. The main technical
tool used in this work involves a scaling property assumed for $\Phi$
\cite{Chen2018} that is a global condition on the L{\'e}vy characteristic
$\Phi(\lambda)$. It is nevertheless difficult to give an interpretation
of this scaling assumption in terms of the subordinator.

The aim of the present work is to extend the class of random times
for which one may obtain information about the time asymptotic of
$v^{E}(x,t)$. We consider the following three classes of admissible
kernels $k\in L_{\mathrm{loc}}^{1}(\mathbb{R}_{+})$, characterized
in terms of the Laplace transforms $\mathcal{K}(\lambda)$ as $\lambda\to0$
(i.e., as local conditions): 
\[
\mathcal{K}(\lambda)\sim\lambda^{\theta-1},\quad0<\theta<1.\tag*{(C1)}
\]
\[
\mathcal{K}(\lambda)\sim\lambda^{-1}L\left(\frac{1}{\lambda}\right),\quad L(y):=\mu(0)\log(y)^{-1}.\tag*{(C2)}
\]
\[
\mathcal{K}(\lambda)\sim\lambda^{-1}L\left(\frac{1}{\lambda}\right),\quad L(y):=C\log(y)^{-1-s},\;s>0,\;C>0.\tag*{(C3)}
\]
We would like to emphasize that these classes of kernels leads to
differential-convolution operators, in particular, the Caputo-Djrbashian
fractional derivative (C1) and distributed order derivatives (C2),
(C2). For each kernel of this type, we establish the asymptotic properties
of the subordinated heat kernels from different classes of a priory
bounds. It is important to stress that in working with much more general
random times (i.e., without the scaling property), a price must be
paid for such an extension, namely the replacement of $v^{E}(x,t)$
by its Cesaro mean. This is the key technical observation that underlies
the analysis of several different model situations.

\section{Preliminaries}

\label{sec:Preliminaries}Let $S=\{S(t),\;t\ge0\}$ be a subordinator,
that is a process with stationary and independent non-negative increments
starting from $0$. They form a special class of L{\'e}vy processes
taking values in $[0,\infty)$ and their sample paths are non-decreasing.
In addition we assume that $S$ has no drift (see \cite{Bertoin96}
for more details). The infinite divisibility of the law of $S$ implies
that its Laplace transform can be expressed in the form 
\[
\mathbb{E}(e^{-\lambda S(t)})=e^{-t\Phi(\lambda)},\quad\lambda\ge0,
\]
where $\Phi:[0,\infty)\longrightarrow[0,\infty)$, called the \emph{Laplace
exponent} (or \emph{cumulant}), is a \emph{Bernstein function}. The
associated L{\'e}vy measure $\sigma$ has support in $[0,\infty)$
and fulfills 
\begin{equation}
\int_{(0,\infty)}(1\wedge\tau)\,d\sigma(\tau)<\infty\label{eq:Levy-condition}
\end{equation}
such that the Laplace exponent $\Phi$ can be expressed as 
\begin{equation}
\Phi(\lambda)=\int_{(0,\infty)}(1-e^{-\lambda\tau})\,d\sigma(\tau),\label{eq:Levy-Khintchine}
\end{equation}
which is known as the \emph{L}{\'e}\emph{vy-Khintchine formula} for
the subordinator $S$. In addition we assume that the L{\'e}vy measure
$\sigma$ satisfy 
\begin{equation}
\sigma(0,\infty)=\infty.\label{eq:Levy-massumption}
\end{equation}
For the given L{\'e}vy measure $\sigma$, we define the function
$k$ by 
\begin{equation}
k:(0,\infty)\longrightarrow(0,\infty),\;t\mapsto k(t):=\sigma\big((t,\infty)\big)\label{eq:k}
\end{equation}
and denote its Laplace transform by $\mathcal{K}$; i.e., for any
$\lambda\ge0$ one has 
\begin{equation}
\mathcal{K}(\lambda):=\int_{0}^{\infty}e^{-\lambda t}k(t)\,dt.\label{eq:LT-k}
\end{equation}
We note that by the Fubini theorem, the function $\mathcal{K}$ is
given in terms of the Laplace exponent. Specifically, 
\[
\mathcal{K}(\lambda)=\int_{0}^{\infty}e^{-\lambda t}\int_{0}^{t}\,d\sigma(s)\,dt=\int_{0}^{\infty}\int_{0}^{s}e^{-\lambda t}\,dt\,d\sigma(s)=\frac{1}{\lambda}\Phi(\lambda),
\]
i.e., 
\begin{equation}
\Phi(\lambda)=\lambda\mathcal{K}(\lambda),\quad\forall\lambda\ge0.\label{eq:Laplace-exponent}
\end{equation}
Given the inverse process $E$ of the subordinator $S$, namely 
\begin{equation}
E(t):=\inf\{s\ge0:\;S(s)\ge t\}=\sup\{s\ge0:\;S(t)\le s\},\label{eq:inverse-sub}
\end{equation}
the marginal density of $E(t)$ will be denoted by $G_{t}(\tau)$,
$t,\tau\ge0$, more explicitly 
\[
G_{t}(\tau)\,d\tau=\partial_{\tau}P(E(t)\le\tau)=\partial_{\tau}P(S(\tau)\ge t)=-\partial_{\tau}P(S(\tau)<t).
\]

\begin{example}[$\theta$-stable subordinator and Gamma processes]
\label{exa:alpha-stable1} 

\begin{enumerate}
\item Let $S$ be a $\theta$-stable subordinator $\theta\in(0,1)$\footnote{The restriction $\theta\in(0,1)$ and not $\theta\in(0,2)$ is due
to the requirement \eqref{eq:Levy-condition}. The boundary $\theta=1$
corresponds to a degenerate case since $S(t)=t$.} with Laplace exponent 
\[
\Phi_{\theta}(\lambda)=\lambda^{\theta}=\frac{\theta}{\Gamma(1-\theta)}\int_{0}^{\infty}(1-e^{-\lambda\tau})\tau^{-1-\theta}\,d\tau,
\]
from which it follows that the L{\'e}vy measure $\sigma$ is given
by 
\[
d\sigma(\tau)=\frac{\theta}{\Gamma(1-\theta)}\tau^{-1-\theta}\,d\tau.
\]
We have $\mathcal{K}(\lambda)=\lambda^{\theta-1}$ and $k(t)={t^{-\theta}}/{\Gamma(1-\theta)}$.
The corresponding GFD $\mathbb{D}_{t}^{(k)}$ is the \emph{Caputo-Djrbashian
fractional derivative} $\mathbb{D}_{t}^{(\theta)}$ of order $\theta\in(0,1)$. 
\item The Gamma process $Y^{(a,b)}$ with parameters $a,b>0$ is given by
its Laplace exponent $\Phi_{(a,b)}$ as 
\[
\Phi_{(a,b)}(\lambda)=a\log\left(1+\frac{\lambda}{b}\right)=\int_{0}^{\infty}(1-e^{-\lambda\tau})a\tau^{-1}e^{-b\tau}\,d\tau,
\]
where the second equality is known as the Frullani integral \cite{Reyna1990}.
The corresponding L{\'e}vy measure is given by 
\[
d\sigma(\tau)=a\tau^{-1}e^{-b\tau}\,d\tau.
\]
We have $\mathcal{K}(\lambda)=\lambda^{-1}a\log\left(1+\frac{\lambda}{b}\right)$
and $k(t)=a\Gamma(0,bt)$. The corresponding GFD is given by 
\[
(\mathbb{D}_{t}^{(a,b)}f)(t)=\frac{d}{dt}\int_{0}^{t}\Gamma(0,b(t-s))(f(s)-f(0))\,ds.
\]
\end{enumerate}
\end{example}

An important characteristic of the density $G_{t}(\tau)$ is given
by its Laplace transform. More precisely, does the $\tau$-Laplace
(or $t$-Laplace) transform of $G_{t}(\tau)$ are known for an arbitrary
subordinator? Thus, we would like to compute the following integrals
\[
\int_{0}^{\infty}e^{-\lambda\tau}G_{t}(\tau)\,d\tau\quad\mathrm{or}\quad\int_{0}^{\infty}e^{-\lambda t}G_{t}(\tau)\,dt.
\]
The answer for the $t$-Laplace transform is affirmative and the result
is given in \eqref{eq:LT-G-t-1} below. On the other hand, for the
$\tau$-Laplace transform a partial answer has been given for the
class of $\theta$-stable processes; namely 
\begin{example}[cf.~Prop.~1(a) in \cite{Bingham1971}]
\label{exa:distr-alphastab-E}If $S$ is a $\theta$-stable process,
then the inverse process $E(t)$ has the Mittag-Leffler distribution,
as follows. 
\begin{equation}
\mathbb{E}(e^{-\lambda E(t)})=\sum_{n=0}^{\infty}\frac{(-\lambda t^{\theta})^{n}}{\Gamma(n\theta+1)}=E_{\theta}(-\lambda t^{\theta}).\label{eq:Laplace-density-alpha}
\end{equation}
It follows from the asymptotic behavior of the Mittag-Leffler function
$E_{\theta}$ that 
\[
\mathbb{E}(e^{-\lambda E(t)})\sim\frac{C}{t^{\theta}},\;\mathrm{as}\;t\to\infty.
\]
In addition, using the fact that 
\begin{equation}
E_{\theta}(-x)=\int_{0}^{\infty}e^{-x\tau}M_{\theta}(\tau)\,d\tau,\quad\forall x\ge0,\label{eq:LT-M-Wright}
\end{equation}
where $M_{\theta}$ is the so-called $M$-Wright (cf.\ \cite{Mainardi_Mura_Pagnini_2010}
for more details and properties), it follows that 
\[
\mathbb{E}(e^{-\lambda E(t)})=\int_{0}^{\infty}e^{-\lambda t^{\theta}\tau}M_{\theta}(\tau)\,d\tau=\int_{0}^{\infty}e^{-\lambda\tau}t^{-\theta}M_{\theta}(\tau t^{-\theta})\,d\tau
\]
from which we obtain the density of $E(t)$ explicitly as 
\begin{equation}
G_{t}(\tau)=t^{-\theta}M_{\theta}(\tau t^{-\theta}).\label{eq:Gdensity-alpha}
\end{equation}
On the other hand, for a general subordinator, the following lemma
determines the $t$-Laplace transform of $G_{t}(\tau)$, with $k$
and $\mathcal{K}$ given in \eqref{eq:k} and \eqref{eq:LT-k}, respectively. 
\end{example}

\begin{lem}
\label{lem:t-LT-G}The $t$-Laplace transform of the density $G_{t}(\tau)$
is given by 
\begin{equation}
\int_{0}^{\infty}e^{-\lambda t}G_{t}(\tau)\,dt=\mathcal{K}(\lambda)e^{-\tau\lambda\mathcal{K}(\lambda)}.\label{eq:LT-G-t}
\end{equation}
In addition, the double ($\tau,t$)-Laplace transform of $G_{t}(\tau)$
is given by 
\[
\int_{0}^{\infty}\int_{0}^{\infty}e^{-p\tau}e^{-\lambda t}G_{t}(\tau)\,dt\,d\tau=\frac{\mathcal{K}(\lambda)}{\lambda\mathcal{K}(\lambda)+p}.
\]
\end{lem}

\begin{proof}
For any $\tau\ge0$ let $\eta_{\tau}$ be the distribution of $S(\tau)$,
that is 
\begin{equation}
\mathbb{E}(e^{-\lambda S(\tau)})=e^{-\tau\Phi(\lambda)}=\int_{0}^{\infty}e^{-\lambda s}\,d\eta_{\tau}(s).\label{eq:Laplace-family}
\end{equation}
Defining 
\begin{equation}
g(\lambda,\tau):=\mathcal{K}(\lambda)e^{-\tau\Phi(\lambda)},\quad\tau,\lambda>0\label{eq:g-t}
\end{equation}
under assumption \eqref{eq:Levy-massumption}, for all $t>0$ it follows
from Theorem 3.1 in \cite{Meerschaert2008} that the density $G_{t}(\tau)$
of the random variable $E(t)$ if given by 
\[
G_{t}(\tau)=\int_{0}^{t}k(t-s)\,d\eta_{\tau}(s).
\]
It follows then that 
\begin{equation}
\int_{0}^{\infty}e^{-\lambda t}G_{t}(\tau)\,dt=g(\lambda,\tau)=\mathcal{K}(\lambda)e^{-\tau\Phi(\lambda)}.\label{eq:LT-G-t-1}
\end{equation}
In fact, by the Fubini's theorem we obtain 
\begin{align*}
\int_{0}^{\infty}e^{-\lambda t}G_{t}(\tau)\,dt & =\int_{0}^{\infty}e^{-\lambda t}\int_{0}^{t}k(t-s)\,d\eta_{\tau}(s)\,dt\\
 & =\int_{0}^{\infty}\int_{s}^{\infty}e^{-\lambda t}k(t-s)\,dt\,d\eta_{\tau}(s)\\
 & =\mathcal{K}(\lambda)\int_{0}^{\infty}e^{-\lambda s}\,d\eta_{\tau}(s)\\
 & =g(\lambda,\tau).
\end{align*}
In addition, it follows easily from \eqref{eq:g-t} that 
\[
\int_{0}^{\infty}g(\lambda,\tau)\,d\tau=\frac{1}{\lambda}
\]
so that \eqref{eq:LT-G-t-1} may be written as 
\[
\int_{0}^{\infty}e^{-\lambda t}\,dt\int_{0}^{\infty}G_{t}(\tau)\,d\tau=\frac{1}{\lambda}
\]
which implies that $G_{t}(\tau)$ is a $\tau$-density on $\mathbb{R}_{+}$:
\[
\int_{0}^{\infty}G_{t}(\tau)\,d\tau=1.
\]
Finally, the double $(\tau,t)$-Laplace transform follows from 
\begin{align}
\int_{0}^{\infty}\int_{0}^{\infty}e^{-p\tau}e^{-\lambda t}G_{t}(\tau)\,dt\,d\tau & =\int_{0}^{\infty}e^{-p\tau}g(\lambda,\tau)\,d\tau\nonumber \\
 & =\mathcal{K}(\lambda)\int_{0}^{\infty}e^{-p\tau}e^{-\tau\lambda\mathcal{K}(\lambda)}\,d\tau\nonumber \\
 & =\frac{\mathcal{K}(\lambda)}{\lambda\mathcal{K}(\lambda)+p}.\qedhere\label{eq:G-double-LT}
\end{align}
\end{proof}

\section{Subordinated Heat Kernel}

In this section, we investigate the long-time behavior of the fundamental
solutions for fractional evolution equations corresponding to random
time changes in the Brownian motion by the inverse process $E(t)$.
We consider three classes of time change, namely those corresponding
to the $\theta$-stable subordinator, the distributed order derivative,
and the class of Stieltjes functions. Henceforth $L$ always denotes
a slowly varying function (SVF) at infinity (see for instance \cite{Bingham1987}
and \cite{Schilling12}), while $C$, $C'$ are constants whose values
are unimportant, and which may change from line to line.

Let $v(x,t)$ be the fundamental solution of the heat equation 
\begin{equation}
\begin{cases}
{\displaystyle \frac{\partial u(x,t)}{\partial t}} & =\frac{1}{2}\Delta u(x,t)\\
u(x,0) & =\delta(x),
\end{cases}\label{eq:Heat-Equation}
\end{equation}
where $\Delta$ denotes the Laplacian in $\mathbb{R}^{d}$. It is
well known that the solution $v(x,t)$ of \eqref{eq:Heat-Equation},
called heat kernel (also known as Green function), is given by 
\begin{equation}
v(x,t)=\frac{1}{(2\pi t)^{\nicefrac{d}{2}}}e^{-\frac{|x|^{2}}{2t}}\label{eq:Heat-Kernel}
\end{equation}
and the associated stochastic process is the classical Brownian motion
in $\mathbb{R}^{d}$. Notice that the heat kernel $v(x,t)$ has the
following long-time behavior 
\[
v(x,t)\sim Ct^{-d/2},\;\mathrm{as}\;t\to\infty.
\]

We are interested in studying the long-time behavior of the subordination
of the solution $v(x,t)$ by the density $G_{t}(\tau)$, that is,
\begin{equation}
v^{E}(x,t):=\int_{0}^{\infty}v(x,\tau)G_{t}(\tau)\,d\tau=\frac{1}{(2\pi)^{\nicefrac{d}{2}}}\int_{0}^{\infty}\tau^{-\nicefrac{d}{2}}e^{-\frac{|x|^{2}}{2\tau}}G_{t}(\tau)\,d\tau.\label{eq:TC-vE1}
\end{equation}
Then $v^{E}(x,t)$ is the fundamental solution of the general fractional
time differential equation, that is, 
\begin{equation}
\begin{cases}
{\displaystyle \mathbb{D}_{t}^{(k)}u(x,t)} & =\frac{1}{2}\Delta u(x,t)\\
u(x,0) & =\delta(x).
\end{cases}\label{eq:FT-Hequation}
\end{equation}
Here $\mathbb{D}_{t}^{(k)}$ are differential-convolution operators
defined, for any nonnegative kernel $k\in L_{\mathrm{loc}}^{1}(\mathbb{R}_{+})$,
by 
\begin{equation}
\big(\mathbb{D}_{t}^{(k)}u\big)(t):=\frac{d}{dt}\int_{0}^{t}k(t-\tau)u(\tau)\,d\tau-k(t)u(0),\;t>0.\label{eq:general-derivative-1}
\end{equation}
(See \cite{Kochubei11} for more details and examples.)

In order to study the time evolution of $v^{E}(x,t)$, one possibility
is to define its Cesaro mean 
\[
M_{t}\big(v^{E}(x,t)\big):=\frac{1}{t}\int_{0}^{t}v^{E}(x,s)\,ds,
\]
which may be written as 
\begin{equation}
M_{t}\big(v^{E}(x,t)\big)=\int_{0}^{\infty}v(x,\tau)M_{t}\big(G_{t}(\tau)\big)d\tau.\label{eq:Cesaro-mean-vE}
\end{equation}
The long-time behavior of the Cesaro mean $M_{t}\big(v^{E}(x,t)\big)$
was investigated in \cite[Sec.~3]{KKdS19} for the three classes of
admissible kernels and $d\ge3$. The method was based on the ratio
Tauberian theorem from \cite{Li2007}. In the next section we use
an alternative method to find the long-time behavior of the Cesaro
mean $M_{t}\big(v^{E}(x,t)\big)$.

\section{Alternative Method for Subordinated Heat Kernel}

\label{sec:Alternative-Method}The Laplace transform method is based
on the result of Lemma~\ref{lem:t-LT-G} wherein the $t$-Laplace
transform of the subordination $v^{E}(x,t)$ is explicitly given by
\begin{align}
(\mathscr{L}v^{E}(x,\cdot))(\lambda) & =C\int_{0}^{\infty}\tau^{-\nicefrac{d}{2}}e^{-\frac{|x|^{2}}{2\tau}}(\mathscr{L}G_{\cdot}(\tau))(\lambda)\,d\tau\nonumber \\
 & =C\mathcal{K}(\lambda)\int_{0}^{\infty}\tau^{-d/2}e^{-\frac{|x|^{2}}{2\tau}-\tau\lambda\mathcal{K}(\lambda)}\,d\tau.\label{eq:LT-vE}
\end{align}
The integral in \eqref{eq:LT-vE} is computed according to the formula,
\[
\int_{0}^{\infty}\tau^{-d/2}e^{-\frac{a}{\tau}-b\tau}\,d\tau=\begin{cases}
\frac{\sqrt{\pi}e^{-2\sqrt{ab}}}{\sqrt{b}}, & d=1,\\
2K_{0}\left(2\sqrt{ab}\right), & d=2,\\
2\left(\frac{a}{b}\right)^{(2-d)/4}K_{d/2-1}\left(2\sqrt{ab}\right), & d\ge3.
\end{cases}
\]
(see for instance \cite[page~146, eqs.~(27), (29)]{erdelyi-I-54}),
where $a=\frac{|x|^{2}}{2}$, $b=\lambda\mathcal{K}(\lambda)$, and
$K_{\nu}(z)$ is the modified Bessel function of the second kind \cite[Sec.~9.6]{AS92}.
The asymptotic of the Bessel function $K_{\nu}(z)$ as $z\to0$ is
well known (e.g., see \cite[Eqs.~(9.6.8) and  (9.6.9)]{AS92}) and
is given by 
\begin{align}
K_{0}(z) & \sim-\ln(z),\label{eq:Kzero-asympt}\\
K_{\nu}(z) & \sim\frac{1}{2}\Gamma(\nu)\left(\frac{z}{2}\right)^{-\nu}\sim Cz^{-\nu},\quad\Re(\nu)>0.\label{eq:K-asympt}
\end{align}
With these explicit formulas, we study each class (C1), (C2), and
(C3) separately. 
\begin{description}
\item [{(C1).}] For this class, $\mathcal{K}(\lambda)=\lambda^{\theta-1}$,
$0<\theta<1$. 

\begin{enumerate}
\item For $d=1$ , we obtain 
\[
(\mathscr{L}v^{E}(x,\cdot))(\lambda)=C\lambda^{-1+\theta/2}e^{-\sqrt{2}|x|\lambda^{\theta/2}}=\lambda^{-(1-\theta/2)}L\left(\frac{1}{\lambda}\right),
\]
where $L(y)=Ce^{-\sqrt{2}|x|y^{-\theta/2}}$ is a SVF. An application
of the Karamata Tauberian theorem(see for example \cite[Sec.~2.2]{Seneta1976}
or \cite[Sec.~1.7]{Bingham1987}) gives 
\[
M_{t}\big(v^{E}(x,t)\big)\sim Ct^{-\theta/2}e^{-\sqrt{2}|x|t^{-\theta/2}}\sim Ct^{-\theta/2},\quad t\to\infty.
\]
\item For $d=2$, we have 
\[
(\mathscr{L}v^{E}(x,\cdot))(\lambda)=C\lambda^{-(1-\theta)}K_{0}(\sqrt{2}|x|\lambda^{\theta/2})=\lambda^{-(1-\theta)}L\left(\frac{1}{\lambda}\right),
\]
where $L(y)=CK_{0}(\sqrt{2}|x|y^{-\theta/2})$ is a SVF. Invoking
the Karamata Tauberian theorem and \eqref{eq:Kzero-asympt} yields,
for $t\to\infty$, 
\[
M_{t}\big(v^{E}(x,t)\big)\sim Ct^{-\theta}K_{0}\big(\sqrt{2}|x|t^{-\theta/2}\big)\sim Ct^{-\theta}\ln\big(\sqrt{2}|x|t^{-\theta/2}\big).
\]
\item For $d\ge3$, the Laplace transform of $v^{E}(x,t)$ has the form
\begin{align*}
(\mathscr{L}v^{E}(x,\cdot))(\lambda) & =C|x|^{(2-d)/2}\lambda^{-(1-\theta)}\left(\frac{1}{\lambda}\right)^{\theta(2-d)/4}K_{\frac{d}{2}-1}(\sqrt{2}|x|\lambda^{\theta/2})\\
 & =\lambda^{-(1-\theta)}L\left(\frac{1}{\lambda}\right),
\end{align*}
where $L(y)=C|x|^{(2-d)/2}y^{\theta(2-d)/4}K_{\frac{d}{2}-1}(\sqrt{2}|x|y^{-\theta/2})$
is a SVF. It follows from the Karamata Tauberian theorem and \eqref{eq:K-asympt}
that 
\[
M_{t}\big(v^{E}(x,t)\big)\sim Ct^{-\theta}L(t)\sim C|x|^{(\theta+1)(2-d)/2}t^{-\theta},\quad t\to\infty.
\]
\end{enumerate}
\item [{(C2).}] Here we have $\mathcal{K}(\lambda)\sim\lambda^{-1}L(\lambda^{-1})$
as $\lambda\to0$, where $L(y)=\mu(0)\log(y)^{-1}$, $\mu(0)\neq0$.
Again we study the three different cases $d=1$, $d=2$ and $d\ge3$.

\begin{enumerate}
\item For $d=1$, the $t$-Laplace transform of $v^{E}(x,t)$ can be written,
for $\lambda\to0$, as 
\begin{align*}
(\mathscr{L}v^{E}(x,\cdot))(\lambda) & =C\lambda^{-1}\log(\lambda^{-1})^{-1/2}e^{-\sqrt{2\mu(0)}|x|\log(\lambda^{-1})^{-1/2}}\\
 & =\lambda^{-1}L\left(\frac{1}{\lambda}\right),
\end{align*}
where $L(y)=C\log(y)^{-1/2}e^{-\sqrt{2\mu(0)}|x|\log(y)^{-1/2}}$
is a SVF. An application of the Karamata Tauberian theorem gives 
\[
M_{t}\big(v^{E}(x,t)\big)\sim CL(t)\sim C\log(t)^{-1/2}e^{-\sqrt{2\mu(0)}|x|\log(t)^{-1/2}},\quad t\to\infty.
\]
\item If $d=2$, we have 
\begin{align*}
(\mathscr{L}v^{E}(x,\cdot))(\lambda) & =C\lambda^{-1}\log(\lambda^{-1})^{-1}K_{0}\left(\sqrt{2\mu(0)}|x|\log(\lambda^{-1})^{-1/2}\right)\\
 & =\lambda^{-1}L\left(\frac{1}{\lambda}\right),
\end{align*}
where $L(y)=C\log(y)^{-1}K_{0}\left(\sqrt{2\mu(0)}|x|\log(y)^{-1/2}\right)$
is a SVF. As $t\to\infty$ then the Karamata Tauberian theorem and
\eqref{eq:Kzero-asympt} we obtain 
\[
M_{t}\big(v^{E}(x,t)\big)\sim CL(t)\sim C\log(t)^{-1}\ln\big(\sqrt{2\mu(0)}|x|\log(t)^{-1/2}\big).
\]
\item For $d\ge3$, it follows that, as $\lambda\to0$, 
\begin{eqnarray*}
(\mathscr{L}v^{E}(x,\cdot))(\lambda) & = & C|x|^{(2-d)/2}\lambda^{-1}\log(\lambda^{-1})^{-1+(2-d)/4}\\
 &  & \times K_{\frac{d}{2}-1}\left(C'|x|\log(\lambda^{-1})^{-1/2}\right)\\
 & = & \lambda^{-1}L\left(\frac{1}{\lambda}\right),
\end{eqnarray*}
where $L(y)=C|x|^{(2-d)/2}\log(y)^{-1+(2-d)/4}K_{\frac{d}{2}-1}\left(C'|x|\log(y)^{-1/2}\right)$
is a SVF. To verify that $L(y)$ is a SVF, one may note that $\log(y)^{-1+(2-d)/4}$
as well as is $K_{\frac{d}{2}-1}\left(C'|x|\log(y)^{-1/2}\right)$
according to \eqref{eq:K-asympt}; the stated result then follows
from Prop.~1.3.6 in \cite{Bingham1987}. It follows from the Karamata
Tauberian theorem and \eqref{eq:K-asympt} that 
\[
M_{t}\big(v^{E}(x,t)\big)\sim CL(t)\sim C|x|^{2-d}\log(t)^{-1},\quad t\to\infty.
\]
\end{enumerate}
\item [{(C3).}] We now have $\mathcal{K}(\lambda)\sim C\lambda^{-1}L(\lambda^{-1})^{-1-s}$,
as $\lambda\to0$ and $s>0$, $C>0$.

\begin{enumerate}
\item For $d=1$, the $t$-Laplace transform of $v^{E}(x,t)$ can be written,
for $\lambda\to0$, as 
\begin{align*}
(\mathscr{L}v^{E}(x,\cdot))(\lambda) & =C\lambda^{-1}\log(\lambda^{-1})^{-(1+s)/2}e^{-C'\sqrt{2}|x|\log(\lambda^{-1})^{-(1+s)/2}}\\
 & =\lambda^{-1}L\left(\frac{1}{\lambda}\right),
\end{align*}
where $L(y)=C\log(y)^{-(1+s)/2}e^{-C'\sqrt{2}|x|\log(y)^{-(1+s)/2}}$
is a SVF, as is easily seen. An application of the Karamata Tauberian
theorem gives, as $t\to\infty$ 
\[
M_{t}\big(v^{E}(x,t)\big)\sim CL(t)\sim C\log(t)^{-(1+s)/2}e^{-C'\sqrt{2}|x|\log(t)^{-(1+s)/2}}.
\]
\item For $d=2$, we have 
\begin{align*}
(\mathscr{L}v^{E}(x,\cdot))(\lambda) & =C\lambda^{-1}\log(\lambda^{-1})^{-1-s}K_{0}\left(C'|x|\log(\lambda^{-1})^{-(1+s)/2}\right)\\
 & =\lambda^{-1}L\left(\frac{1}{\lambda}\right),
\end{align*}
where 
\[
L(y)=C\log(y)^{-1-s}K_{0}\left(C'|x|\log(y)^{-(1+s)/2}\right)
\]
is a SVF. Use the Karamata Tauberian theorem and \eqref{eq:Kzero-asympt}
now yield the behavior as $t\to\infty$ 
\[
M_{t}\big(v^{E}(x,t)\big)\sim CL(t)\sim C\log(t)^{-1-s}\ln\big(C'|x|\log(t)^{-(1+s)/2}\big).
\]
\item For $d\ge3$, it follows that 
\begin{eqnarray*}
(\mathscr{L}v^{E}(x,\cdot))(\lambda) & = & C|x|^{(2-d)/2}\lambda^{-1}\log(\lambda^{-1})^{-(1+s)(1-(2-d)/4)}\\
 &  & \times K_{\frac{d}{2}-1}\left(C'\sqrt{2}|x|\log(\lambda^{-1})^{-(1+s)/2}\right)\\
 & = & \lambda^{-1}L\left(\frac{1}{\lambda}\right),
\end{eqnarray*}
as $t\to\infty$, where 
\[
L(y)=C|x|^{(2-d)/2}\log(y)^{-(1+s)(2+d)/4}K_{\frac{d}{2}-1}\left(C'\sqrt{2}|x|\log(y)^{-(1+s)/2}\right)
\]
is a SVF. We note that $L(y)$ is the product of two SVF's which is
a SVF (see Prop.~1.3.6 in \cite{Bingham1987}). It the follows from
the Karamata Tauberian theorem and \eqref{eq:K-asympt} that 
\[
M_{t}\big(v^{E}(x,t)\big)\sim CL(t)\sim C|x|^{2-d}\log(t)^{-1-s},\quad t\to\infty.
\]
\end{enumerate}
\end{description}
\begin{rem}[Gaussian convolution kernel]
\label{rem:Convolution-kernel}We consider the nonlocal operator
$\mathcal{L}$ on functions $u:\mathbb{R}^{d}\longrightarrow\mathbb{R}$
defined in integral form by 
\begin{equation}
(\mathcal{L}u)(x):=(a*u)(x)-u(x)=\int_{\mathbb{R}^{d}}a(x-y)[u(y)-u(x)]\,dy,\label{eq:operator}
\end{equation}
where the convolution kernel $a$ is non-negative, symmetric, bounded,
and integrable, i.e., 
\begin{equation}
a(x)\ge0,\qquad a(x)=a(-x),\qquad a(x)\in L^{\infty}(\mathbb{R}^{d})\cap L^{1}(\mathbb{R}^{d}).\label{eq:condition-a1}
\end{equation}
In addition, the kernel $a$ is a density in $\mathbb{R}^{d}$ with
finite second moment; explicitly 
\begin{equation}
\langle a\rangle:=\int_{\mathbb{R}^{d}}a(x)\,dx=1,\quad\int_{\mathbb{R}^{d}}|x|^{2}a(x)\,dx<\infty.\label{eq:condition-a2}
\end{equation}
Since $\mathcal{L}$ is a bounded operator in $L^{2}(\mathbb{R}^{d})$,
its heat semigroup $e^{t\mathcal{L}}$ can be easily computed by using
the exponential series according to 
\[
e^{t\mathcal{L}}=e^{-t}e^{ta\ast}=e^{-t}\sum_{k=0}^{\infty}t^{k}\frac{a^{\ast k}}{k!}=e^{-t}\mathrm{Id}+e^{-t}\sum_{k=1}^{\infty}t^{k}\frac{a^{\ast k}}{k!}.
\]
By removing the singular part $e^{-t}\mathrm{Id}$ of the heat semigroup,
we obtain the \emph{regularized} heat kernel 
\begin{equation}
v(x,t)=e^{-t}\sum_{k=1}^{\infty}t^{k}\frac{a^{\ast k}\left(x\right)}{k!}\label{eq:v-NL-CP}
\end{equation}
with the source at the origin. In other words, for any $f\in L^{2}(\mathbb{R}^{d})$,
a solution to the nonlocal Cauchy problem 
\begin{equation}
\begin{cases}
{\displaystyle \frac{\partial u(x,t)}{\partial t}} & =\mathcal{L}u(x,t),\\
u(x,0) & =f(x),
\end{cases}\label{eq:NL-CP}
\end{equation}
has the form $u(x,t)=e^{-t}f(x)+(v\ast f)(x,t)$ with $v$ given by
\eqref{eq:v-NL-CP}. In particular, the fundamental solution of the
problem \eqref{eq:NL-CP} is 
\[
u(x,t)=e^{-t}\delta(x)+v(x,t).
\]
For any $r>0$, if $|x|\le rt^{\nicefrac{1}{2}}$, then the following
asymptotic for $v(x,t)$ holds as $t\to\infty$ 
\begin{equation}
v(x,t)=\frac{1}{(4\pi t)^{\nicefrac{d}{2}}}e^{-\frac{|x|^{2}}{4t}}(1+o(t^{-\nicefrac{1}{4}})),\label{eq:asymp-v-gaussian-a}
\end{equation}
see \cite[Thm.~2.1]{Grigoryan2018}. If we denote by $v^{E}(x,t)$
the subordination of $v(x,t)$ by the density $G_{t}(\tau)$, then
the Cesaro mean of $v^{E}(x,t)$ has long time behavior given as above
for the different classes of admissible kernels $k$. 

\end{rem}

\subsection*{Acknowledgement}

Financial support from FCT--Funda{\c c\~a}o para a Ci{\^e}ncia
e a Tecnologia through the project UID/MAT/04674/2019 (CIMA Universidade
da Madeira) is gratefully acknowledged. The work of the first-named
author was funded in part by the budget program of Ukraine No. 6541230
``Support to the development of priority research trends''. It was
also supported in part under the research work ``Markov evolutions
in real and p-adic spaces\textquotedbl{} of the Dragomanov National
Pedagogical University of Ukraine. Finally, we would like to thanks
the anonymous referee of the first draft of the paper to point out
some misprints, valuable comments and suggestions that have led to
a significant improvement of the paper.

\end{document}